\documentclass[12pt]{elsarticle}

\usepackage{graphicx}

\usepackage{theorem}
\theoremstyle{break}

\newtheorem{theorem}{Theorem}[section]

\theoremstyle{plain}

\newcommand{\dd}{\mathrm{d}}

\newcommand{\kk}{\mathbf{k}}
\newcommand{\ii}{\mathrm{i}}
\newcommand{\e}{\mathrm{e}}

\newcommand{\oo}{\mathrm{o}}

\if@compatibility

\fi

\newcommand{\re}{\mathop{\mathrm{Re}}}
\newcommand{\im}{\mathop{\mathrm{Im}}}

\usepackage{amsmath,amssymb,amscd,upgreek}

\newcommand\qedsymbol{\hbox{\rlap{$\sqcap$}$\sqcup$}}
\newenvironment{proof}{\textbf{Proof.}}{\\ \mbox{ }\hfill\qedsymbol\\}

\renewcommand{\vec}[1]{\mathbf{#1}}

\renewcommand{\im}{\mathop{\mathrm{Im}}}
\renewcommand{\re}{\mathop{\mathrm{Re}}}
\newcommand{\loc}{\mathrm{loc}}
\newcommand{\D}{\mathrm{D}}
\newcommand{\I}{\mathrm{I}}

\newcommand{\upi}{\uppi}

\begin{document}

\begin{frontmatter}

\author[MS]{M. Seredy\'{n}ska\corref{CA}}
\address[MS]{
Institute of Fundamental Technological Research,
Polish Academy of Sciences\\
 ul. Pawi\'{n}skiego 5b,
02-106 Warszawa, PL}
\ead{msered@ippt.gov.pl}

\author[AH]{Andrzej Hanyga}
\address[AH]{
ul. Bitwy Warszawskiej 14 m. 52, 
02-366 Warszawa, PL}

\ead{ajhbergen@yahoo.com}

\cortext[CA]{Corresponding Author}

\date{30.04.2010}

\title{Theory of attenuation and finite propagation speed in 
viscoelastic media}

\bibliographystyle{elsarticle-num}

\begin{abstract}
It is shown that the dispersion and attenuation functions in a linear 
viscoelastic medium 
with a positive relaxation spectrum 
can be expressed in terms of a positive measure. Both functions have a sublinear 
growth rate at very high frequencies. In the case of power law attenuation
positive relaxation spectrum ensures finite propagation speed. For more general
attenuation functions the requirement of finite propagation speed imposes a 
more 
stringent condition on the high-frequency behavior of attenuation. 
It is demonstrated that superlinear power law frequency dependence of attenuation 
is incompatible with finite speed of 
propagation and with the assumption of positive relaxation spectrum. 
\end{abstract}


\end{frontmatter}

\pagebreak

\noindent\textbf{List of symbols.}\vspace{0.2cm}\\
\begin{tabular}{lll}
$[a,b]$ & $\{ x \mid a \leq x \leq b \}$ & for $a < b < \infty$;\\
$]a,b]$ & $\{ x \mid a < x \leq b \}$ & for $-\infty \leq a < b < \infty$;\\
$]a,b[$ & $\{ x \mid a < x  < b \}$ & for any $-\infty \leq a <  b \leq \infty$;\\ 
$\mathbb{R}_+$ & $]0,\infty[$ & the set of positive reals $x$; \\
$\mathcal{F}_x(f)$ & $\int_{-\infty}^\infty \e^{\ii \vec{k}\cdot x} \, f(x) \, \dd x$ & Fourier transform;\\
$\mathcal{L}_t(f)$ & $\int_0^\infty \e^{-p t} f(t)\, \dd t$ & Laplace transform;\\
$\tilde{f}(p)$ & $\int_0^\infty f(t) \, \e^{-p t} \, \dd t$ & Laplace transform;\\
$\theta(t)$ &  & unit step function;\\
$t_+^\alpha$ & $\theta(t)\, \vert t \vert^\alpha$ & homogeneous distribution;\\
$\D^\alpha_{\mathrm{C}}$ & $\I^{n-\alpha}\, \D^n$ & Caputo fractional derivative,\\
$f \sim_a g$ &  $0 < \lim_{x\rightarrow a} \vert f/g \vert < \infty $;\\
$f = \oo[g]$ & $\lim_{x\rightarrow a} \vert f/g \vert = 0$ &at $a =0$ or $\infty$
\end{tabular}

\section{Introduction.}

Viscoelastic media of the Newtonian, hereditary or combined type are known to have 
positive relaxation spectra \citep{Molinari,Day70a,Day70b,AnderssenLoy02a,HanDuality}. 
Any network of springs and dashpots combined in parallel and serially has this property
\cite{HanISIMM}. This is also true of many 
theoretical models of polymers (the Rouse model \citep{Rouse53}, Zimm's reptation-based 
models \cite{Zimm56}, 
Generalized Gaussian models \citep{GurtovenkoBlumen05}), rubber \citep{Kuhn,LublinerPanoskaltsis}, metals \citep{Andrade2} and rocks \citep{MinsterAnderson}. Universality of this property 
has led to same attempts at explaining it by a general thermodynamic or "fading memory" principle 
\citep{Day70a,BerisEdwards93,AnderssenLoy02a}).
The same property is shared by some (but not all) dielectric media 
\citep{HanSerDielectrics07}. 

Positive relaxation spectrum has some consequences for wave dispersion and attenuation 
in such media.
We shall use the term "admissible" to refer to the attenuation and 
dispersion functions which are compatible with positive relaxation time spectrum.
It will be shown that every admissible pair of attenuation and dispersion functions
has an integral representation in terms of a positive Radon measure, which we call 
the dispersion-attenuation measure. The measure is arbitrary except for a 
very mild growth condition. The two expressions taken together can be considered 
as a dispersion relation in parametric form. 

In contrast the acoustic Kramers-Kronig (K-K) dispersion relations \cite{WeaverPao81} 
are non-local. The K-K dispersion relations express the dispersion function in terms 
of the attenuation function or 
conversely. This presupposes that one of these functions (usually the 
dispersion function) is accurately known and and is admissible. 
On the other hand, substituting any positive 
measure in the parametric dispersion relation yields an admissible
pair of dispersion and attenuation functions. 

There is no a priori argument in favor of the K-K dispersion relations in acoustics. 
In electromagnetic theory the K-K relations follow from the causality of the time-domain 
kernel appearing in the constitutive equations of a dispersive dielectric medium.
In acoustics the K-K relations are derived from an ad hoc assumption about the analytic 
properties of the wave number. It is namely assumed that the wave number is the 
Fourier transform of a causal function (or distribution) $L$. This assumption 
is not justified unless a physical meaning is given to the causal function $L$.
Attempts to give a physical meaning to the function $L$ in terms of an ad hoc wave 
equation are discussed in Sec.~\ref{sec:superlinear}. The wave equation is however 
not consistent with the constitutive equations of viscoelasticity unless an 
approximation are made. 

It will be also shown that in media with positive 
relaxation spectrum the frequency dependence of the attenuation function 
$\mathcal{A}_1(\omega)$
in the high frequency range is sublinear: $\mathcal{A}_1(\omega) = \oo[\omega]$. 
This fact raises some problems which we also try to address in this paper. 
In the case of power law attenuation 
the attenuation and dispersion are proportional to a power of frequency 
$\vert \omega\vert^\alpha$. Numerous experiments in acoustics indicate
that the power law accurately represents the frequency dependence
of dispersion and attenuation over several decades of frequency. 
If the power law extends to arbitrarily large frequencies then 
viscoelastic theory based on positive relaxation spectrum 
implies that the exponent of the power law satisfies the inequalities 
$0 \leq \alpha < 1$.  The same inequality follows from the requirement 
of finite speed of propagation. 

Investigations of creep and relaxation 
in viscoelastic materials 
always support the assumption of positive relaxation spectrum 
(e.g. \cite{Andrade1,Andrade2} for creep in metals, \cite{MinsterAnderson}
for the upper mantle) and therefore models derived 
from constitutive relations exhibit sublinear attenuation and dispersion 
at high frequencies. In particular this applies to the Cole-Cole, 
Havriliak-Negami, Cole-Davidson and Kohlrausch-Williams-Watts relaxation 
laws commonly applied in phenomenological rock mechanics, polymer
rheology \cite{BagleyTorvik4}, bio-tissue mechanics (e.g. for bone 
collagen \cite{SasakiAl93}) as well as for ionic glasses \cite{CariniAl84}.
It also applies to the most common constitutive model of bio-tissue 
based on power law creep\cite{DespratAl05} or to the soft-glassy 
model of cytoskeleton mechanics\cite{FabryAl2001,FabryAl03}.

Sublinear power law attenuation 
($\alpha < 1$) has also been reported in experimental investigations 
involving PVC and lava samples 
\cite{RibodettiHanygaGJI}. Experimental ultrasound investigations of 
numerous materials point however to higher
values of the exponent in the power law attenuation. For example in 
ultrasound investigations of soft 
tissues the exponent varies between 1 and 1.5, while in some viscoelastic fluids
such as castor oil it lies between 1.5 and 2. In multi-walled carbon nanotubes 
the exponent lies above 1.1 \cite{MobleyAll2009}.
We shall refer to this case as superlinear frequency dependence.
Typical values of the power law exponent in medical applications using 
ultrasound transducers 
are $\alpha = 1.3$ in bovine liver for $1\div 100~\mathrm{MHz}$, $\alpha = 1\div 2$ in
human myocardium and other bio-tissues \cite{SzaboWu00}. Values in the range $1\div 2$
are observed at lower frequencies in aromatic polyurethanes 
\cite{GuessCampbell95}. Nearly linear frequency dependence of attenuation 
is well documented in seismology \cite{Futterman}. Approximately 
linear frequency dependence of attenuation has been observed in 
geological materials in the range 140~Hz to 2.2~MHz. 

Several papers have been
devoted to a theoretical justification of the superlinear dispersion-attenuation 
models 
\cite{Szabo1,Szabo2,ChenHolm03,WatersHughersBrandenburgerMiller,KellyMcGoughMeerschaert08,CobboldSushilovWeathermon04,KellyMcGoughMeerschaert08}.
In order to resolve some problems Chen and Holm \cite{ChenHolm04} suggested to add a fractional Laplacian of the velocity field 
of order $y$, where $0 < y < 2$,  to the usual Laplacian of the displacement 
field in the equations of motion.  Their approach still leads to an unbounded sound speed for superlinear attenuation 
and adds a new problem: their wave equation cannot be derived from a viscoelastic 
constitutive equation. The underlying physical model of Chen and Holm's equation is 
for the time being unclear.

The main factor affecting viscoelastic behavior and viscoelastic wave propagation is relaxation represented by the relaxation modulus. This remains true for fractal media such as polymers and bio-tissues. In contrast to anomalous diffusion, viscoelasticity is incompatible with fractional Laplacians. 

Superlinear power law attenuation is incompatible with a finite upper limit
on propagation speed. This also applies to more general attenuation functions
with a linear or superlinear asymptotic growth in the high frequency range.
In fact we shall see that the boundary between finite and infinite speed 
of propagation corresponds to asymptotic attenuation $\sim a\, \omega/\ln(\omega)$.  
An abnormal feature of wave propagation in media with superlinear 
power laws (and more generally in media with superlinear asymptotic 
frequency dependence) is appearance of precursors. The precursors extend 
to infinity and thus the speed of propagation of disturbances is infinite. 

It is a challenging problem how to explain the incompatibility 
between the theory of viscoelasticity (or finite propagation speed) and experiment 
in the case of superlinear attenuation. 
It seems likely that the attenuation observed at ultrasound frequencies should 
be associated with an intermediate frequency range in which the attenuation 
function behaves
in a significantly different way from its asymptotic behavior. However even in 
this case problems remain. 
A superlinear power law cannot be strictly valid in any finite range 
of frequencies. Since the local exponent 
$\alpha_1(\omega) := \ln \mathcal{A}_1(\omega)
/\ln \omega$ is analytic for $\omega > 1$, it cannot be constant over 
a finite frequency range without being constant over the entire real line.
Hence a local superlinear power law entails superlinear asymptotic behavior 
which is incompatible with the requirements of positive relaxation time spectrum and 
finite propagation speed. 

\section{Constitutive assumptions and basic definitions.}

In viscoelasticity the relaxation modulus $G$, 
defined by the constitutive stress-strain relation
\begin{equation} \label{eq:0}
\upsigma(t) = \int_0^t G(t-s) \, \dot{e}(s)\, \dd s + \zeta \, \dot{e}(t)
\end{equation}
is assumed to have positive relaxation spectral density $h$. The latter 
statement means that for every $t > 0$ 
\begin{equation} \label{eq:1}
G(t) = G_0 + \int_0^\infty \e^{-t r}\, h(r)\, \dd \ln(r), \qquad t > 0
\end{equation}
where $r = 1/\tau$ is the inverse of the relaxation time and $h(r) \geq 0$. 
Eq.~\eqref{eq:1} represents 
a superposition of a continuum of Debye elements. For mathematical convenience 
eq.~\eqref{eq:1} will be replaced by a more general equation
\begin{equation} \label{eq:2}
G(t) = \int_{[0,\infty[\;} \e^{-t r}\,\mu(\dd r), \qquad t > 0
\end{equation} 
where $\mu$ is a positive measure:
$\mu([a,b]) \geq 0$ for every interval $[a,b]$ of the positive real axis. 
As indicated in the subscript of the integral sign the range of integration 
is the set of reals satisfying the inequality $0 \leq r < \infty$. In general 
the measure $\mu(\{0\})$ of the one-point set $\{0\}$ is finite and equal
to the equilibrium modulus $G_\infty := \lim_{t\rightarrow\infty} G(t)$. 
An additional assumption
\begin{equation}
\int_{[0,\infty[\;} \frac{\mu(\dd r)}{1 + r} < \infty
\end{equation} 
ensures that the function $G$ is integrable over the interval $[0,1]$ 
\cite{HanDuality}. The relaxation modulus assumes 
a finite value $G(0) = M$ at 0 if the measure $\mu$ has a finite mass $M = \mu([0,\infty[)$.

In contrast to eq.~\eqref{eq:1} the integral representation \eqref{eq:2} 
includes as special cases finite spectra of relaxation times 
corresponding to superpositions of a finite number of Debye elements 
\begin{equation} \label{eq:4}
G(t) = \sum_{n=1}^N c_n \, \e^{-r_n\, t}, \qquad c_n > 0, r_n \geq 0 
\quad\text{for $n = 1,\ldots,N$}
\end{equation}
(Prony sums), infinite discrete spectra corresponding to Dirichlet series as well 
as discrete spectra embedded in continuous spectra.
However the main advantage of \eqref{eq:2} over \eqref{eq:1} is the 
availability of a very powerful mathematical theory which ensures logical 
equivalence of certain statements about material response functions. 

In particular a function satisfying \eqref{eq:2} is completely monotone. 
A function $G(t)$ is
said to be \emph{completely monotone} if it continuously differentiable to every order and
\begin{equation} \label{eq:CM}
(-1)^n \, \D^n \, G(t) \geq 0 \quad\text{for all non-negative integers 
$n$ and $t > 0$}
\end{equation}
Bernstein's theorem \cite{WidderLT,Jacob01I} 
asserts that eq.~\eqref{eq:2} is {\em equivalent} to \eqref{eq:CM}.
There is no such simple characterization of functions which have the integral
representation \eqref{eq:1}. 
Hence the requirement of positive relaxation time spectrum \eqref{eq:2} is
equivalent to the complete monotonicity of the relaxation 
modulus \cite{Bland:VE,Molinari,AnderssenLoy02b,HanDuality}.

For us the main benefit from using \eqref{eq:2} instead of \eqref{eq:1} 
is the equivalence 
of \eqref{eq:2} with a property of the dispersion and attenuation that will be
explained below. That is, certain statements about attenuation and dispersion 
functions follow from \eqref{eq:2} or, conversely, imply that \eqref{eq:2}
does not hold.

We now recall the definition of a Bernstein function \cite{BergForst}. A 
differentiable 
function $f$ on $]0,\infty[$ is a Bernstein function if $f \geq 0$ and 
its derivative 
$f^\prime$ is completely monotone. A Bernstein function is non-negative, 
continuous 
on $]0,\infty[$ and non-decreasing, hence it has a finite value at 0. A 
function $f$ on $[0,\infty[$ which has the form $f(x) = x^2\, \tilde{g}(x)$ 
for some Bernstein function $g$ is called a \emph{complete Bernstein function} (CBF) 
\cite{Jacob01I}. It can be 
proved that every complete Bernstein function is a Bernstein function
\cite{Jacob01I}. 

The following facts about complete Bernstein functions will be needed here. 
\begin{theorem} \label{thm:J}
A real function $f$ on $]0,\infty[$ is a complete Bernstein function if (i) $f$ 
has an analytic continuation $f(z)$
to the complex plane cut along the negative real axis; (ii) $f(0) \geq 0$, (iii) 
$f\left(\overline{z}\right) = \overline{f(z)}$, (iv) $\im f(z) \geq 0$ in the
upper half plane $\im z \geq 0$.
\end{theorem}
Theorem~\ref{thm:J} has the following corollaries\cite{HanSer09}: 
(1) If $f$ is a CBF then $f^\alpha$ is a CBF if $0 \leq \alpha \leq 1$. \\
(2) If $f \not\equiv 0$ then the function $x/f(x)$ is a CBF. 

Every complete Bernstein function $f$ has the integral representation:
\begin{equation} \label{eq:40}
f(x) = a + b\, x + x \int_{]0,\infty[\;} \frac{\nu(\dd r)}{x + r}
\end{equation}
where $a, b \geq 0$ and $\nu$ is a positive measure satisfying the inequality
\begin{equation} \label{eq:5}
\int_{]0,\infty[\;} \frac{\nu(\dd r)}{1 + r} < \infty
\end{equation}
\cite{Jacob01I}.
Note that the integration domain does not include the point 0. The contribution 
of 0 would have the form $\nu(\{0\}) \, x$, hence it could be included in 
the second term on the right-hand side of \eqref{eq:40}. If the point 0 is excluded
then the constants $a, b$ and the measure $\nu$ are uniquely defined by the
function $f$. 

If $G$ is a completely monotone function integrable over $[0,1]$ and 
satisfying eq.~\eqref{eq:2} then 
$$\tilde{G}(p) = \int_{[0,\infty[\;} \frac{\mu(\dd r)}{p + r} =
\frac{\mu(\{0\})}{p} + \int_{]0,\infty[\;} \frac{\mu(\dd r)}{p + r}$$
It follows that
\begin{equation} \label{eq:8} 
\mathcal{Q}(p) := p \, \tilde{G}(p) = \mu(\{0\}) + p \int_{]0,\infty[\;} 
\frac{\mu(\dd r)}{p + r}
\end{equation}
is a complete Bernstein function \cite{HanSer09}.

A function $f$ is called a Stieltjes function if it has an integral representation
\begin{equation}
f(x) = a + \int_{]0,\infty[\;} \frac{\nu(\dd r)}{x + r}
\end{equation}
where $a \geq 0$ and $\nu$ is a positive measure satisfying \eqref{eq:5}.

\section{Dispersion and attenuation.}
\label{sec:DA}

The Green's function $\mathcal{G}(t,x)$ is defined as the solution of 
the problem
\begin{equation}
\rho u_{,tt} = G(t)\ast u_{,txx} + \delta(x)\, \delta(t)\\
\end{equation}
with zero initial data.
In a three-dimensional space 
\begin{multline}  \label{eq:3Dgreen}
\mathcal{G}(t,x) = \\ \frac{-1}{(2 \upi)^3\, r} 
\int_{-\ii \infty + \varepsilon}^{\ii \infty + \varepsilon} \e^{p t} \, 
\frac{1}{\mathcal{Q}(p)}\,\dd p \int_{-\infty}^\infty \, 
\frac{\e^{\ii k r}}{k^2 + B(p)^2}\, k \, \dd k = \\ 
\frac{1}{8 \upi^2 \,\ii\, r} \int_{-\ii \infty + \varepsilon}^
{\ii \infty + \varepsilon} \,\frac{1}{\mathcal{Q}(p)}
\e^{p\, t - B(p) \, r}\, \dd p 
\end{multline}
where
\begin{equation} \label{eq:6} 
B(p) := p\,\frac{\rho^{1/2}}{\mathcal{Q}(p)^{1/2}} 
\end{equation}
It will be assumed that $G$ is completely monotone and that the equilibrium modulus
$G_\infty := \lim_{t \rightarrow \infty} G(t)0$ is positive. The first assumption 
entails that 
$\mathcal{Q}$ is a complete Bernstein function. From the corollaries of 
Theorem~\ref{thm:J} we know that the functions 
$\mathcal{Q}^{1/2}$ and $p/\mathcal{Q}(p)^{1/2}$ are complete Bernstein functions. 
Hence $B$ is a complete Bernstein function. The second assumption entails that 
$\mathcal{Q}(0) = \lim_{p\rightarrow 0} 
\left[ p \,\tilde{G}(p) \right] = G_\infty > 0$, hence $B(0) = 0$. In view of
\eqref{eq:40} the function $B$ has a representation 
$B(p) = c\, p + b(p)$ where $c \geq 0$, the dispersion-attenuation function 
$b$ has the integral representation
\begin{equation} \label{eq:attenuationspectrum}
b(p) := p\, \int_{]0,\infty[\;} \frac{\nu(\dd r)}{p + r}
\end{equation}
and $\nu$ is a positive measure satisfying the inequality \eqref{eq:5}. 
A function $b$ having an integral representation \eqref{eq:attenuationspectrum}
will be called an \emph{admissible dispersion-attenuation function}.
The measure $\nu$ is the spectral measure of the dispersion-attenuation function $b$. 
A dispersion-attenuation function $b$ of a viscoelastic material with a positive 
relaxation spectrum is admissible. The converse statement is not true, as will be 
shown later.

We shall now prove that  
$\lim_{p\rightarrow\infty} B(p)/p = c$.   The last conclusion follows from the
fact that for $p > 1$ 
$$\frac{1}{p + r} \leq \frac{1}{1 + r}$$
and the right-hand side is integrable with respect to the measure $\nu$ in view of  
eq.~\eqref{eq:5}. For $p \rightarrow \infty$ the integrand of 
$$\int_{]0,\infty[\;} \frac{\nu(\dd r)}{p + r}$$
tends to zero, hence, by the Lebesgue Dominated Convergence Theorem, 
the integral tends to zero as well. It follows additionally that 
$b(p) = \mathrm{o}[p]$ for $p \rightarrow \infty$. We now note that
$\lim_{p\rightarrow\infty} \mathcal{Q}(p) = G_0 := \lim_{t\rightarrow 0} G(t)$
is the elastic (instantaneous response) modulus. Hence $c = 1/c_0$, where
$c_0 := \sqrt{G_0/\rho}$. 
Hence eq.~\eqref{eq:3Dgreen} can be recast in a more explicit form 
\begin{equation} \label{eq:Green3D}
\mathcal{G}(t,x) = \frac{1}{8 \upi^2 \,\ii\, r} \int_{-\ii \infty + \varepsilon}^
{\ii \infty + \varepsilon} \,\frac{1}{\mathcal{Q}(p)}
\e^{p\,( t - \vert x \vert/c_0) - b(p) \, r}\, \dd p 
\end{equation}
and the dispersion-attenuation function $b(p)$ has sublinear growth.

We also note that the attenuation function is non-negative in the right-half
of the complex $p$-plane: 
\begin{equation}
\mathcal{A}(p) := \re b(p) = \int_{]0,\infty[\;} 
\frac{\vert p \vert^2 + r \, \re p}
{\vert p + r\vert^2} \nu(\dd r) \geq 0 
\end{equation}
for $\re p \geq 0$. 

The derivative 
\begin{equation}
b^\prime(p) = \frac{b(p)}{p} - \int_{]0,\infty[\;} 
\frac{\nu(\dd r)}{(p + r)^2} \leq \frac{b(p)}{p}
\end{equation}
Moreover 
$$b^\prime(p) = \int_{]0,\infty[\;} \frac{r\, \nu(\dd r)}{(p + r)^2} 
\geq 0$$
On account of \eqref{eq:5} the function $r/(p + r)^2 \leq  
1/(1 + r)$ (for $p > 1$) is integrable with respect to the measure $\nu$
and tends to zero for $p \rightarrow \infty$.  
Hence $b^\prime(p) = \mathrm{o}[1]$ for $p \rightarrow \infty$.

The attenuation function $\mathcal{A}$ and the dispersion function 
$\mathcal{D}(p) :=
-\im b(p)$ satisfy linear dispersion equations in parametric form:
\begin{gather}
\mathcal{A}(p) = \int_{]0,\infty[\;} 
\frac{\vert p \vert^2 + r \, \re p}
{\vert p + r \vert^2} \nu(\dd r)\\
\mathcal{D}(p) = -\im p \int_{]0,\infty[\;} 
\frac{r}{\vert p + r\vert^2} \nu(\dd r)
\end{gather}

The measure $\nu$ represents the dispersion-attenuation spectrum. An elementary 
dispersion-attenuation is represented by the function $p/(p + r)$ for a fixed 
value of $r$. Given a CBF function $b$, the ratio $g(p) := b(p)/p$ is a Stieltjes 
function. If the dispersion-attenuation measure $\nu(\dd r) = n(r)\, \dd r$, 
$n \in \mathcal{L}^1_\loc(\mathbb{R}_+)$,
then the {\em dispersion-attenuation density} $n$ can be calculated by applying 
the formula for inversion of the Stieltjes transform:
\begin{equation} \label{eq:inverseattenuation}
n(r) = \frac{1}{\upi} \lim_{\varepsilon \rightarrow 0+} \,
\im g(-r - \ii \varepsilon) 
\end{equation}
\cite{HanSerPRSA}. More generally, $\im g(p) \geq 0$ if $\im p \leq 0$.
The distributional limit $\lim_{\varepsilon \rightarrow 0+} \,
\im g(-r - \ii \varepsilon)$ is non-negative, hence it is a non-negative measure. 
For a general non-negative measure $\nu$ we thus have the inversion formula
\begin{equation}
\nu([a,b]) = \int_{[a,b]} \im \lim_{\varepsilon \rightarrow 0+} \,
g(-r - \ii \varepsilon)
\end{equation}
for an arbitrary interval $[a,b]$ of the real axis.

Eq.~\eqref{eq:inverseattenuation} can be used to determine the
dispersion-attenuation measure when the dispersion-attenuation function is given
in the form of an analytic function. On the other hand either eq.~\eqref{eq:xxx} 
or eq.~\eqref{eq:xxy} can be used to determine the spectral density $h_1$. This function
is subject to only two constraints: (1) it is non-negative; (2) it must decrease 
faster that $1/\ln(\tau)$ as $\tau \rightarrow \infty$. 

The Green's function can be calculated by an inverse Laplace transform which 
involves 
integration over the imaginary axis $p = -\ii \omega$. It is therefore important to 
investigate the functions $b, \mathcal{D}$ and $\mathcal{A}$ on the imaginary axis in
the $p$ plane. Substituting $p = -\ii \omega$ in the integral representation of $b$
we get the following frequency-domain formulas:
\begin{gather}
\mathcal{A}_1(\omega) := \re b(-\ii \omega) = \omega^2 \int_{]0,\infty[\;} 
\frac{\nu(\dd r)}{\omega^2 + r^2} \label{eq:frAtt}\\
\mathcal{D}_1(\omega) := -\im b(-\ii \omega) = \omega \int_{]0,\infty[\;} 
\frac{r\, \nu(\dd r)}{\omega^2 + r^2} \label{eq:frDisp}
\end{gather}

For completeness we also note that the propagation speed $c(\omega)$ and the
quality factor $Q(\omega)$ can be expressed in terms of the measure $\nu$ 
\begin{equation}
\frac{1}{c(\omega)} = -\im B(-\ii\omega)/\omega \equiv \frac{1}{c_0} + 
\frac{\mathcal{D}_1(\omega)}{\omega} 
\end{equation}
\begin{equation}
Q(\omega) = \frac{\omega}{4 \upi \mathcal{A}_1(\omega)} \left( \frac{1}{c_0} 
+ \mathcal{D}_1(\omega)\right)
\end{equation}

It can now be proved that the attenuation is sublinear as a function of real frequency.
\begin{theorem} \label{thm:attf}
 As $\omega \rightarrow \infty$
\begin{enumerate} 
\item $\mathcal{A}_1(\omega) = \oo[\omega]$ and  $\mathcal{D}_1(\omega) = \oo[\omega]$;\\
\item if $M := \int_{]0,\infty[\;} \nu(\dd r) < \infty$ then $\mathcal{A}_1(\omega) 
\rightarrow M$.
\end{enumerate}
\end{theorem}
\begin{proof}
For $\omega > 1$ 
$$\frac{1}{2} \frac{\omega}{\omega^2 + r^2} \leq \frac{\omega}{(\omega + r)^2} 
\leq \frac{1}{1 + r} \frac{\omega}{\omega + r} \leq \frac{1}{1 + r}$$
The integrand of $\mathcal{A}_1(\omega)/\omega$ is thus bounded by an integrable function
and it tends to 0 for $\omega \rightarrow \infty$. The thesis follows from
the Lebesgue Dominated Convergence Theorem and eq.~\eqref{eq:5}. \index{Lebesgue Dominated Convergence Theorem}

For $M = \int_{]0,\infty[\;} \nu(\dd r) < \infty$ we note that the integrand of $\mathcal{A}_1(\omega)$
tends to 1 and thus $\mathcal{A}_1(\omega) \rightarrow M$ by the Lebesgue Dominated Convergence Theorem. 

For $\omega > 1$ the integrand $r/(\omega^2 + r^2)$ of $\mathcal{D}_1(\omega)/\omega$ is bounded by an integrable function
$$\frac{1}{2} \frac{r}{\omega^2 + r^2} \leq \frac{r}{(\omega + r)^2}
\leq \frac{r}{1 + r} \frac{1}{1 +  r} \leq \frac{1}{1 +  r}$$
Since the integrand tends to 0 for $\omega \rightarrow \infty$, $\mathcal{D}_1(\omega)/\omega = \oo[1]$.
\end{proof}
By similar arguments one can prove that for $\omega \rightarrow 0$
\begin{gather}
\mathcal{A}_1(\omega) = \oo\left[\omega^2\right]\\
\mathcal{D}_1(\omega) = \oo\left[\omega^2\right]
\end{gather}

Since $x/(1 + x)$ is an increasing function, eq.~\eqref{eq:frAtt} implies that the
attenuation function $\mathcal{A}_1(\omega)$ is non-decreasing function of 
$\omega > 0$.

If $\nu$ is absolutely continuous with respect to the logarithmic scale on the
$r$ axis then
there is a locally integrable function $h(r)$ such that $\nu(\dd r) = h(r) \dd r/r$
and $\int_0^\infty [h(r)/(1 + r)] \dd r/r < \infty$. The function $b$ now assumes the form
$b(p) = p \int_0^\infty h(r)/[p + r] \, \dd r$. The variable $r$ can be replaced by 
a variable $\tau = 1/r$ with the  dimension of time. Let $h_1(\tau) :=
h(1/\tau)$. Substituting $p = -\ii \omega$ as well we get a formula used in 
physical acoustics
\begin{equation} \label{eq:xxx}
\mathcal{A}_1(\omega) = \omega^2 \int_0^\infty 
\frac{\tau \,h_1(\tau)}{1 + \tau^2 \, \omega^2} \dd \tau
\end{equation}
and its companion formula 
\begin{equation} \label{eq:xxy}
\mathcal{D}_1(\omega) = \omega \int_0^\infty 
\frac{h_1(\tau)}{1 + \tau^2 \, \omega^2} \dd \tau 
\end{equation}
where $h_1 \geq 0$ and
\begin{equation} \label{eq:xxz}
\int_0^\infty \frac{h_1(\tau)}{1 + \tau} \dd \tau < \infty
\end{equation}

The variable $\tau$ is different
from the relaxation time $\tau = 1/r$ appearing in eq.~\eqref{eq:2}. For example 
in the case of a finite relaxation spectrum $\tau_n = 1/r_n$ (eq.~\eqref{eq:4})
the attenuation spectrum is also discrete but different from the relaxation spectrum.
The attenuation spectrum can be calculated from the relaxation spectrum by solving 
a set of non-linear equations following from the equation 
\begin{equation} \label{eq:QB}
\mathcal{Q}(p) \, B(p)^2 = \rho \, p^2
\end{equation} and conversely.  
 
\section{Implications of finite speed of wave propagation.}

Since $1/\mathcal{Q}(p) = p \, \tilde{J}(p) = J(0) + \widetilde{J^\prime}(p)$,
where $J$ is the creep compliance, the Green's function \eqref{eq:Green3D} 
can be expressed in the form of a convolution
$$\mathcal{G}(t,x) = J^\prime(t)\ast H(t-\vert x \vert/c_0,\vert x\vert)
+ J_0 \, H(t-\vert x \vert/c_0,\vert x\vert)$$
where $H(\tau,r)$ is the inverse Laplace transform of the function
$\e^{-b(p)\,r}/(4 \upi r)$. In terms of the inverse Fourier 
transformation
$$H(\tau,r) = \frac{1}{8 \upi^2 \, r} \int_{-\infty}^\infty
\e^{-\ii \omega \tau - b(-\ii \omega)\, r} \, \dd \omega$$
Note that $b(-\ii \omega) = \overline{b(\ii \omega)}$.
The integrand is square integrable if $r > 0$ and 
\begin{equation} \label{eq:condPW}
\int_0^\infty \e^{-2 \re b(-\ii \omega)\, r} \dd \omega < \infty
\end{equation}

If eq.~\eqref{eq:condPW} holds then 
the Paley-Wiener theorem (Theorem~XII in \cite{PaleyWiener}) can be applied.
By this theorem  
$H(\tau,r) = 0$ for $\tau < 0$
if and only if 
\begin{equation} \label{eq:PW}
\int_{-\infty}^\infty \frac{\re b(-\ii \omega)}{1 + \omega^2} \dd \omega < \infty
\end{equation}
The function $J^\prime$ is causal. Hence, 
if $H(\tau,r) = 0$ for $\tau < 0$ then 
$$\mathcal{G}(t,x) = \int_0^{t-\vert x \vert/c_0} 
J^\prime(s) \, H(t - \vert x \vert/c_0 - s, \vert x \vert) \, \dd s $$
vanishes for $t < \vert x \vert/c_0$.

\begin{figure}
\includegraphics[width=0.75\textwidth]{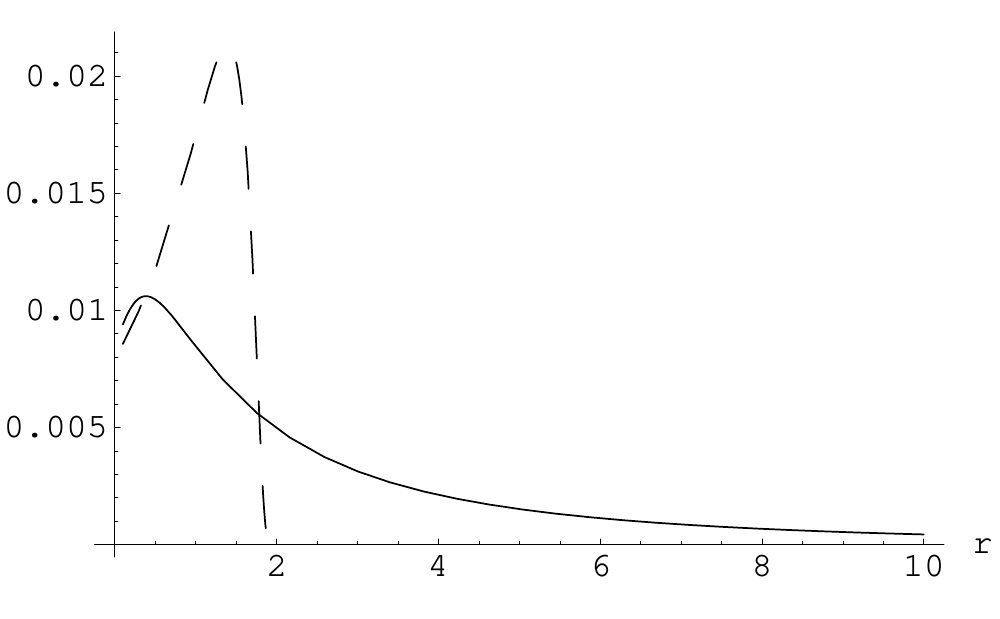}
\caption{Two snapshots of Green's functions for power law attenuation with 
$\alpha=3/2$ (solid line) and $\alpha=1/2$ (dashed line).} 
\label{fig:1}
\end{figure}

In particular, if $b(p) \sim_\infty a\, p^\alpha$, $a > 0$, $\re p \geq 0$,
the inequalities \eqref{eq:condPW} and \eqref{eq:PW} are satisfied if 
$0 < \alpha < 1$. 

It follows from Theorem~\ref{thm:attf} that the asymptotic growth of 
of the attenuation function $\re b(-\ii \omega)$ in the high-frequency range is sublinear.
This is however insufficient to satisfy inequality \eqref{eq:PW}
and vanishing of the wave field ahead of the wavefront $\vert x \vert = c_0\, t$. 
For example for $b(-\ii \omega) \sim_\infty  \omega/\ln^\alpha(\omega)$ 
the integral in \eqref{eq:PW} does not always converge. Since we are interested in
the high-frequency behavior, we can replace $1 + \omega^2$ by $\omega^2$ in the
denominator of \eqref{eq:PW}. 
For $\alpha = 1$ 
$$\int_\e^\omega \frac{\dd y}{y\, \ln(y)} = \ln(\ln(\omega))$$
is unbounded for $\omega \rightarrow \infty$, hence \eqref{eq:PW} is not satisfied. 
For $\alpha \neq 1$
$$\int_\e^\omega \frac{\dd y}{y\, \ln^\alpha(y)} = 
\frac{\ln^{1-\alpha}(\omega)-1}{1 - \alpha}$$
is unbounded for $\alpha < 1$ and bounded for $\alpha > 1$. 
We thus see that in contrast to the power law attenuation, linear growth is not a limit
case for finite propagation speed. Some sublinear cases will also exhibit precursors 
ahead of the wave front. 

In particular the function $b(p) = p/\ln^\alpha(1 + p)$ with 
$0 \leq \alpha \leq 1$ is a CBF with the 
asymptotic properties discussed in the previous paragraph. Indeed,
for $\im p \geq 0$ the argument $\psi = \arg(1 + p)$ satisfies the inequality
$0 \leq \psi \leq \upi$. Hence $\ln(1 + p)$ maps the upper half plane into 
itself and is
non-negative for $p \geq 0$. Hence, by Theorem~\ref{thm:J}, $\ln(1 + p)$ is a CBF.
The same is true for $\ln^\alpha(1 + p)$ if $0 < \alpha \leq 1$.  In view of 
a property of CBF functions mentioned after Theorem~\ref{thm:J} this implies
that $b(p)$ is a CBF if $0 < \alpha \leq 1$. Moreover $b(0) = 0$ and
$\lim_{p\rightarrow \infty} b(p)/p = 0$, hence $b(p)$ is an admissible 
dispersion-attenuation function. 
We thus have produced an example of an admissible dispersion-attenuation 
function $b(p)$ such that $b(-\ii \omega)$ has sublinear growth in the
high frequency range but \eqref{eq:PW} is not satisfied.

\section{Examples of dispersion-attenuation functions.}

\subsection{Power-law attenuation.}

Power-law attenuation is commonly used to match experimental dispersion and attenuation 
data for a wide variety of real viscoelastic materials such as polymers, bio-tissues 
and some viscoelastic fluids. 

Consider the viscoelastic medium defined by the following equation of motion
\begin{equation} \label{eq:Rok}
\rho\, \left(\D_\mathrm{C}^2 + 2 a\, \D_\mathrm{C}^{1+\alpha} + 
a^2\, \D_\mathrm{C}^{2 \alpha} \right) u = A\, \nabla^2 \, u + \delta(t)\, \delta(x)
\end{equation}
with the initial data $u(0,x) = u_{,t}(0,x) = 0$ in $d$ dimensions, $d =1,3$.
$\D_\mathrm{C}^\alpha$ denotes the Caputo fractional derivative of order 
$\alpha$ 
\begin{equation}
\D_\mathrm{C}^\alpha f(x) := \int_0^t \frac{(t-s)^{-\alpha}}{\Gamma(1-\alpha)} 
f^\prime(s) \, \dd s
\end{equation}
\cite{PodlubnyBook}.

It is assumed that $A > 0$, $a \geq 0$, $0 < \alpha < 1$.
The Laplace transformation $\mathcal{L}_t$ with respect to the time variable and the 
Fourier transformation $\mathcal{F}_{x}$
with respect to the spatial variable bring eq.~\eqref{eq:Rok} to the following form
\begin{equation} \label{eq:Rok1}
\rho\,g(p)^2 \, \hat{u}(p,\kk) 
= -A\,\kk^2\,\hat{u}(p,\kk) + 1 
\end{equation}
where $\hat{u} := \mathcal{F}_{x}\left(\mathcal{L}_t(u)\right)$ and 
$g(p) := p + a\, p^\alpha$. Eq.~\eqref{eq:Rok1} implies that in the case of power law attenuation 
$\mathcal{Q}(p) = p\,A/g(p)^2$ and $B(p) = g(p)/c_0$.

We shall begin with solving eq.~\eqref{eq:Rok} in one dimension. Applying the inverse
Fourier transformation to eq.~\eqref{eq:Rok1}:
$$\tilde{u}^{(1)}(p,x) = \frac{1}{2 \upi} \int_{-\infty}^\infty 
\frac{\e^{\ii k x + p t}}{\rho \,g(p)^2 + A\, k^2} \dd k$$
The contour can be closed by a large half-circle in upper-half complex $k$-plane if 
$x > 0$, in the lower- half complex $k$-plane if $x < 0$ and 
$$\frac{1}{A\, k^2 + \rho\,g(p)^2} = \frac{1}{2 \ii g(p)\, A} 
\left[\frac{1}{k - \ii g(p)/c_0} - \frac{1}{k + \ii g(p)/c_0}  \right]$$
where $c_0 := \sqrt{A/\rho}$. We now restrict ourselves to imaginary values of 
$p = -\ii w$. Since $\im [\ii\, g(-\ii w)] = \re g(-\ii w) = a\,\vert w \vert^\alpha \,
\cos((1-\alpha)\,\upi/2) \geq 0$, the residuum at $\pm \ii g(p)/c_0$ contributes if 
$\pm x > 0$. Hence
$$\tilde{u}(p,x) = \frac{1}{2 A} 
\frac{\e^{-g(p) \vert x \vert/c_0 + p t}}{g(p)} $$
and
\begin{equation}
u^{(1)}(t,x) = \frac{1}{4 \upi \ii \, A} \int_{-\ii \infty}^{\ii \infty} \frac{1}{g(p)}
\e^{p \,t - g(p)\vert x \vert/c_0} \dd p
\end{equation}
The solution $u^{(3)}$ of \eqref{eq:Rok} in a three-dimensional space 
is given by the formula
$$u^{(3)}(t,x) = \frac{-1}{2 \upi r} \frac{\partial}{\partial r} 
u^{(1)}(t,r)$$
Hence 
\begin{multline}
u^{(3)}(t,x) = \frac{1}{8 \upi^2 \ii \vert x \vert\, c_0\, A}  
\int_{-\ii \infty}^{\ii \infty} \e^{p \,(t - \vert x \vert/c_0) - 
a \, p^\alpha\, \vert x \vert} \dd p = \\
\frac{1}{4 \upi a^{1/\alpha}\,\vert x \vert^{1+1/\alpha}\, c_0 \, A}\, 
P_\alpha\left((t - \vert x \vert/c_0)/(a \vert x \vert)^\alpha\right)
\end{multline}
where
\begin{equation}
P_\alpha(z) := \frac{1}{2 \upi \ii} \int_{-\ii \infty}^{\ii \infty} 
\e^{y z} \, \e^{-y^\alpha}\, \dd y
\end{equation}
The function $P_\alpha$ is a totally skewed L\'{e}vy stable probability density 
\cite{Sato:Levy,Zolotarev2}. 

For $\alpha = 1/2$ the $\alpha$-stable probability $P_\alpha$
can be expressed in terms of elementary functions. For $\alpha = 1/3, 2/3$ the
$\alpha$-stable probability can be expressed in terms of Airy functions, see 
Ref.\cite{HanQAM}. The Green's functions are easy to calculate explicitly
in these cases. The Green's function vanishes at and outside
the wavefront, as can be seen in the case $\alpha  1/2$ in Fig.~\ref{fig:1}. 
For $\alpha \geq 1$  there is no wavefront and the signal 
peak is preceded by a precursor extending to infinity, as can be seen 
in the case of $\alpha = 3/2$ in 
Fig.~\ref{fig:1}. The limit case, as we shall see,
is the asymptotic behavior $b(p) \sim_\infty a\,p/\ln(p)$. In the limit case 
the propagation speed is unbounded.

Since  $\re b(-\ii \omega) =
a\, \omega^\alpha \, \cos(\upi \alpha/2) > 0$ and for $p = -\ii \omega$ 
$$\left\vert \e^{p\, t - g(p)\, r/c_0} \right\vert = 
\e^{-a\,\omega^\alpha\, \cos(\alpha\,\upi/2)}$$
the integrals
$$f(t,r) := \frac{1}{2 \upi \ii} \int_{-\infty}^\infty
 (-1)^n \,p^{-2}\,p^m\, g(p)^{n+1} \, 
\e^{p\,t - g(p) \, r/c_0} \, \dd p$$
are uniformly convergent for $r > 0$  and 
a theorem in \cite{WhittakerWatson}, Sec.~4.44 implies that
 the derivatives
$\partial^{n+m}\, f/\partial t^n\, \partial r^m$ exist for all positive integers 
$n, m$ and $r > 0$. For $0 < \alpha < 1$ the Green's function vanishes outside
the wavefront and is infinitely differentiable everywhere. Hence the Green's function
vanishes at the wavefront with all its derivatives \cite{HanQAM}.

We shall conclude the section with a spectral characterization of 
the power law attenuation: 
\begin{theorem}
The spectral measure of the dispersion-attenuation function $a\, p^\alpha$, 
$0 \leq \alpha < 1$, is 
\begin{equation} \label{eq:attspepower}
\nu(\dd \xi) = a\,\frac{\sin(\upi \alpha)}{\upi} \xi^{\alpha-1} \, \dd \xi
\end{equation}
\end{theorem}
\begin{proof}
Substituting the Laplace transform 
\begin{equation}
y^{-\alpha} =  \int_0^\infty \e^{-y z} \, 
\left[z^{\alpha-1}/\Gamma(\alpha)\right]\, \dd z
\end{equation}
in the Laplace transform
\begin{equation}
x^{\alpha-1} = \int_0^\infty \e^{-x y} \, 
\left[y^{-\alpha}/\Gamma(1-\alpha)\right]\, \dd y
\end{equation}
yields the Stieltjes transform
$$x^{\alpha - 1} = \frac{\sin(\alpha \upi)}{\upi} 
\int_0^\infty \frac{z^{\alpha-1}}{x + z} \, \dd z$$
which implies eq.~\eqref{eq:attspepower}.
\end{proof}

\subsection{Eq.~\eqref{eq:attenuationspectrum} does not imply that the 
relaxation modulus is completely monotone.}

All the properties of the dispersion-attenuation function $b$ derived in 
Sec.~\ref{sec:DA} follow from the fact that $B(p) = p\,\rho^{1/2}/\mathcal{Q}(p)^{1/2}$ 
is a complete
Bernstein function. It follows from the above property of $B(p)$ that $\mathcal{Q}(p)^{1/2}$ is a 
complete Bernstein function\cite{HanSer09}. On the other hand $\mathcal{Q}$ need not be a complete Bernstein function.
Therefore $G$ need not be a completely monotone function even though $B$ is
a complete Bernstein function. 

An example is provided 
by the power law attenuation model. The function $B$ is a complete Bernstein function
for $0 \leq \alpha < 1$. On the other hand $\mathcal{Q}$ is a complete Bernstein function
and $G$ is completely monotone if and only if $1/ \leq \alpha < 1$.
(Fig.~\ref{fig:0}). 
\begin{figure}
\begin{center}
\includegraphics[width=0.75\linewidth]{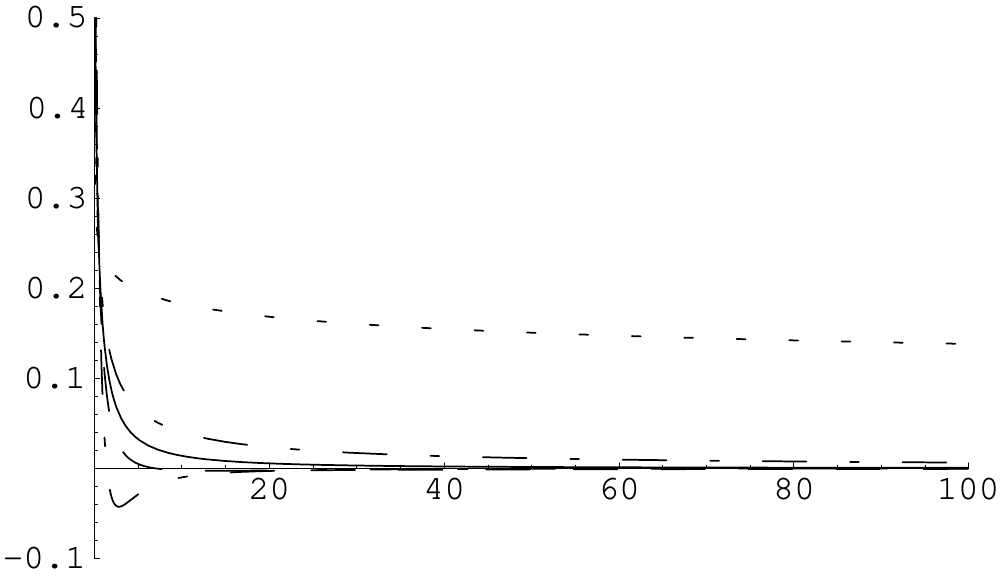}
\end{center}
\caption{The relaxation modulus $G$ for the power law attenuation. 
The exponent values are $\alpha = 0.3, 0.4, 0.5, 0.6, 0.9$,
from bottom to top.}
\label{fig:0}
\end{figure}
Another example is given in Section~\ref{ssec:discrete}..

\subsection{Cytoskeleton rheology.}

The continuum model of cytoskeleton rheology \cite{Mofrad09,DespratAl05} is based on a simple creep compliance in pure extension 
$J(t) = J_1 \, t^\beta/\Gamma(1+\beta)$ with $J_1 > 0$, $0 \leq \beta < 1$. The last inequality follows
from the fact that the creep compliance in a viscoelastic material with a positive
relaxation time spectrum is a Bernstein function\cite{HanDuality}. 
A typical value is $\beta = 0.24$. It follows 
that $\mathcal{Q}(p) = J_1^{\;-1} p^\beta$, $c_0 = \infty$, $B(p) = b(p) = (\rho\, J_1)^{1/2}\,
p^\alpha$ with $\alpha := 1 - \beta/2 < 1$. This model is also confirmed by oscillatory experiments which closely match the storage modulus  
$G^\prime(\omega) = J_1^{-1}\, \omega^\beta\, \cos(\upi \beta/2)$ and loss modulus 
$G^{\prime\prime}(\omega) = J_1^{-1}\, \omega^\beta\, \sin(\upi \beta/2)$ at 
$\omega \cong 1\,\mathrm{Hz}$. In this case $G(t) = G_1\,t^{-\beta}/\Gamma(1-\beta)$, 
where $G_1 = 1/J_1$.

The fractional Kelvin-Voigt model of bio-tissue \cite{ZhangAl07} is defined by 
the relaxation modulus 
\begin{equation} \label{eq:10}
G(t) = G_\infty + G_1 \, t^{-\beta}/\Gamma(1-\beta)
\end{equation}
with $G_\infty, G_1 \geq 0$ and $0 \leq \beta < 1$. 
The last inequality 
follows from the fact that the relaxation modulus 
of a viscoelastic medium with a positive is a completely monotone function. 
In this case $\mathcal{Q}(p) = G_\infty + G_1\, p^\beta$, $c_0 = \infty$ and
$B(p) = b(p) \sim_0 \left[\rho/G_1\right]^{\;1/2}\, p^\alpha$, where 
$\alpha := 1 - \beta/2 < 1$. For $t \rightarrow 0$ the fractional 
Kelvin-Voigt relaxation modulus is asymptotically equivalent to the relaxation 
modulus of the Rouse model of dilute
polymer solutions \citep{Rouse53} and to the relaxation modulus of a  derived in 
\cite{KellyMcGough09} for a ladder model of soft tissue.

In all these cases the propagation speed is unbounded because of an infinite value of
the relaxation modulus at $t = 0$. In all these cases 
the attenuation function follows a sublinear power law.

\subsection{Discrete dispersion-attenuation spectra.}
\label{ssec:discrete}

If the measure $\nu = \sum_{n=1}^N c_n\, \varepsilon_{r_n}$, $r_n, c_n > 0$,
where $\varepsilon_a = \delta(r - a)$ denotes the Dirac measure concentrated at the 
point $a$, then
\begin{equation} \label{eq:7}
b(p) = p \sum_{n=1}^N \frac{c_n}{p + r_n}
\end{equation} 

It is however possible to prove that \eqref{eq:7} implies that $\mathcal{Q}$ is not 
a complete Bernstein function. The proof is not very short although quite elementary.
Therefore eq.~\eqref{eq:7} implies that the relaxation modulus $G$ is not
completely monotone and the relaxation time spectrum is not positive. It follows
that the dispersion-attenuation spectra of viscoelastic media always have 
a non-trivial continuous component. 

\section{Superlinear power law attenuation and the K-K dispersion relations.}
\label{sec:superlinear}

A large number of papers in ultrasound acoustics deal with superlinear power 
law attenuation \cite{Szabo1,Szabo2,SzaboWu00,ChenHolm03,UrbanGreenleaf09}. 
In order to express attenuation in terms of measured dispersion they resort to the 
K-K dispersion relations for the wave number $K(\omega) := 
\ii B(-\ii \omega)$. It is also assumed that a model of acoustic propagation is 
correct (causal) if the K-K relations with some number of subtractions 
are satisfied. 

The acoustic K-K dispersion relations follow from the assumption 
that the reduced wavenumber function $K(\omega) - \omega/c_0$ is the Fourier 
transform of a causal function (or, more generally,
causal tempered distribution) $L(t)$. A priori the function $L(t)$ has no 
physical meaning and the assumption of causality of $L$ is unwarranted.  
A physical meaning can be given to the function $L$ by the assumption
that the equation of motion can be expressed in the following form: 
\begin{equation} \label{eq:motion}
c_0^{\;-2}\, u_{,tt} + (2/c_0)\, L\ast u_{,t} + L\ast L\ast u =
u_{,xx}
\end{equation}
Here 
$$\tilde{L}(p) = B(p) - p/c_0 $$
Causality of the integro-differential operator requires that $L$ is a causal function.
In this case $K(\omega) =  
\omega/c_0 + \ii \, \tilde{L}(-\ii \omega) = 
\ii \omega/c_0 + \hat{L}(\omega)$.  
Eq.~\eqref{eq:motion} is however
incompatible with the viscoelastic constitutive equations. In a viscoelastic 
equation of motion integral operators should act on the Laplacian of $u$.
This kind of causality does not have any physical meaning. 
It is more important that the Green's function is non-causal in the
superlinear case.

In \cite{SzaboWu00} the authors try to guess the viscoelastic constitutive
equation 
compatible with \eqref{eq:motion}. Their approach involves an approximation
of a spatial derivative by a temporal derivative. It is however possible 
to avoid
an approximation by shifting the dispersive terms on the left-hand side 
of \eqref{eq:motion} to the right-hand side. The Laplace transform of the 
left-hand side of eq.~\eqref{eq:motion} is
$$ p^2\, \left[ 1 + c_0\, \tilde{L}(p)/p \right]^2/c_0^{\;2}$$
assuming that $u(0,x) = u_{,t}(0,x) = 0$. 
Hence \eqref{eq:motion} has the form 
$$c_0^{\;-2}\,u_{,tt} = \left[G(t)\ast u_{,tx}\right]_{,x}$$
where $G$ is the inverse Laplace transform of
\begin{equation} \label{eq:xxxa}
p^{-1}\,\left[ 1 + c_0\, \tilde{L}(p)/p \right]^{-2} 
\end{equation}
Expression~\eqref{eq:xxxa} is the Laplace transform of a completely monotone
function if and only if $\left[1 + c_0\, \tilde{L}(p)/p\right]^2/p$ is the 
Laplace transform of a Bernstein function $f$ \cite{HanDuality,HanSerPRSA}. 

For $\tilde{L}(p) = a\, p^\alpha$, with $a, \alpha > 0$, the function $f$ 
assumes the form
$\theta(t) + 2 c_0 \,a\, t_+^{1-\alpha}/\Gamma(1-\alpha) + c_0^2\,a^2\, 
t_+^{2 ( 1-\alpha)}/\Gamma(3-2\alpha)$.
It is obvious that $f$ is a Bernstein function (and $G$ is completely monotone)
if and only if $1/2 \leq \alpha < 1$. We also note that 
$L(t) = a\, t_+^{-\alpha-1}/\Gamma(-\alpha)$ is a distribution.
Convolution with 
$L(t)$  is a Riemann-Liouville fractional differential operator of order 
$\alpha$. 
The order of the fractional differential equation \eqref{eq:motion} 
is $2$ if $\alpha \leq 1$ and $2 \alpha $ if $\alpha > 1$. (In the literature
the highest-order derivative $L\ast L\ast u$ is often incorrectly neglected).
For $\alpha > 1$ the order of the time-differential operator exceeds
the order of the spatial differential operator, hence the equation is 
formally parabolic. 

Superlinear power law attenuation results in a rather strange behavior 
of the dispersion function. 
Note that the attenuation function 
\begin{equation}
\mathcal{A}(\omega) := \re B(-\ii \omega) = a \, \cos(\alpha\, \upi/2) \,
\vert \omega \vert^\alpha
\end{equation} 
is non-negative. Hence, if $0 < \alpha \leq 1$ then the constant $a$ must be 
non-negative. On the other hand, for  
$1 < \alpha < 3$ the constant $a$ must be non-positive. 
As the frequency tends to infinity the frequency-dependent phase  
speed $c(\omega)$, given by the equation  
$1/c(\omega) := \re [\ii B(- \ii \omega)]/\omega = 1/c_0 + 
a \, \vert \omega \vert^{\alpha-1} \, \sin(\alpha\,\upi/2)$
increases to its maximum value $c_0$ if $0 \leq \alpha < 1$ 
("abnormal dispersion"). 
For $1 \leq \alpha < 2$ it increases from $c_0$ at zero frequency to infinity
at a finite frequency 
$\omega_1 = 
\left[ c_0 \, \vert a \vert \, \sin(\alpha\, \upi/2) \right]^{1/(\alpha-1)}$ 
and changes sign. 
For $2 < \alpha < 3$ phase speed decreases from $c_0$ at $\omega = 0$
to 0 at infinite frequency ("normal dispersion"). This behavior of the
phase speed results in the appearance of precursors.

\section{Viscoelastic constitutive equations with a Newtonian viscosity component.}

We shall briefly consider viscoelastic constitutive equations with 
a Newtonian component:
\begin{equation} \label{eq:01}
\upsigma(t) = \int_0^t G(t-s) \, \dot{e}(s)\, \dd s + \eta \, \dot{e}(t)
\end{equation}
where the Newtonian viscosity $\eta > 0$. If $G$ is completely monotone then 
the function 
$\mathcal{Q}(p) = p \, \tilde{G}(p) + \eta\, p$ is
a complete Bernstein function and the theory developed above remains valid except that 
$\lim_{p\rightarrow \infty} \mathcal{Q}(p) = \infty$.  

A particular case is
the hysteretic damping model defined by
the constitutive equation \eqref{eq:01} and eq.~\eqref{eq:10} \cite{FabryAl2001}. 

Since the limits $\lim_{p\rightarrow 0} \,[p \, \tilde{G}(p)] = G_\infty$ and $\lim_{p\rightarrow\infty}\, [p \, \tilde{G}(p)] = G_0$ are finite, 
$B(0) = 0$ 
and $B(p) \sim_\infty (\rho/\eta)^{1/2} \, p^{1/2}$. Thus 
$c = 0$ and $B(p) \equiv b(p)$. This implies that the Green's function does not
have any wave fronts. 
Furthermore it follows from the theory of asymptotics of Stieltjes transforms that 
the dispersion-attenuation measure has an asymptotic behavior 
consistent with the asymptotics of $b(p)$:
\begin{equation}
\nu(\dd r) \sim_0 \frac{(\rho/\eta)^{1/2}}{\upi} \, r^{-1/2} \, \dd r
\end{equation}
Hence high frequency asymptotics of the attenuation and dispersion functions is entirely determined by the Newtonian component:
\begin{equation}
\mathcal{A}_1(\omega) \sim_\infty \frac{1}{\upi}\, (\rho/\eta)^{1/2} \, \omega^2 \int_0^\infty \frac{r^{-1/2}}{r^2 + \omega^2} \dd r = \\  \frac{2}{\upi}\, (\rho/\eta)^{1/2} \,\omega^{1/2}
\end{equation}
Similarly
\begin{equation}
\mathcal{D}_1(\omega) \sim_\infty \frac{1}{\upi}\, (\rho/\eta)^{1/2} \, \omega \int \frac{r^{1/2}}{r^2 + \omega^2} \dd r = \frac{2}{3\,\upi} (\rho/\eta)^{1/2} \,\omega^{1/2}
\end{equation}

The low frequency asymptotics of the attenuation and dispersion is entirely different. 
For small $p$ the term $p \, \tilde{G}(p) \sim_0 G_\infty > 0$
$$B(p) \sim_0 p\, \left[1 + \eta\, p/G_\infty\right]^{-1/2}/c_\infty \sim_0   
p \, \left[1 - \eta\, p/(2 G_\infty)\right]/c_\infty$$
where $c_\infty := \sqrt{G_0/\rho}$.
Hence
$$B(-\ii \omega) = -\ii \omega/c_\infty + \beta(-\ii \omega)$$
where
\begin{equation} \label{eq:asymp}
\beta(-\ii \omega) \sim_0 \frac{\eta}{2 G_\infty\, c_\infty} \omega^2
\end{equation}
This shows that in the low frequency range the attenuation function
has an approximately quadratic frequency dependence.

The above described low- and high-frequency behavior described here has 
been experimentally 
confirmed for superionic glasses, with the high-frequency asymptotic exponent 
$\alpha \approx 0.8$ \cite{CariniAl84}.

The inverse Fourier transform \eqref{eq:3Dgreen} involves the exponential
$$\e^{-\ii \omega\, (t - r/c_\infty)} \, \e^{-\beta(-\ii \omega) \, r}$$
For large $t - r/c_\infty$ eq.~\eqref{eq:asymp} can be substituted in
\eqref{eq:3Dgreen}. Upon the changing
the integration variable $y = \omega\, r^{1/2}$ and setting
$\mathcal{Q}(p) \approx \mathcal{Q}(0) = G_\infty$ in the integrand eq.~\eqref{eq:3Dgreen} assumes the form
\begin{equation} \label{eq:lra}
\mathcal{G}(t,x) = f\left((t-r/c_\infty)/r^{1/2}\right)/r^{3/2}
\end{equation}
where
\begin{equation}
f(z) := \frac{1}{8 \upi^2 \, G_\infty} \int_{-\infty} \e^{-\ii y z} \e^{-\eta\, y^2/(2 G_\infty\, c_\infty)} \, \dd y = 
\frac{\sqrt{2} \, c_\infty^{\;1/2}}{8 \upi^{3/2} \, G_\infty^{\;1/2} \, \eta^{1/2}}\,\e^{- \eta\, z^2/(8 G_\infty\, c_\infty)}\,
\end{equation}
According to eq.~\eqref{eq:lra}  the Green's function in the long-range 
assumes the form of a Gaussian pulse traveling with the speed $c_\infty$. The quadratic 
dependence of the attenuation can be associated with the $r^{1/2}$ scaling of
the pulse width.   

A rigorous theory of long-range asymptotics of viscoelastic Green's functions in the
absence of Newtonian viscosity is presented in \cite{HanAsympVE}. In these cases 
the the pulse form and the scaling of the pulse width are different.

\section{Experimentally measured superlinear power law attenuation in 
 an intermediate frequency range.}

The restrictions imposed on the dissipation-attenuation exponent $\alpha$ by 
the positive relaxation spectrum and by finite propagation speed apply to 
the asymptotic behavior of $b(p)$ at infinity. At low frequencies the exponent 
is close to 2. It is likely that the experimentally observed superlinear power law 
behavior applies to an intermediate frequency range. 

For limited frequency bands very good fits of viscoelastic models 
to superlinear power laws can be obtained. For example, the Voigt model 
\begin{equation}
\sigma = G_0 \, e + \eta \, \dot{e} 
\end{equation}
is a special case of the class of viscoelastic models considered here. In this case
\begin{gather}
\mathcal{A}_1(\omega) = \sqrt{\frac{\rho \, \omega^2\,\left(\sqrt{G_0^{\;2} 
+ \eta^2\, \omega^2}-G_0\right)}{2 (G_0^{\;2} + \eta^2 \, \omega^2)}}\\
c(\omega) = \sqrt{\frac{2 \left( G_0^{\;2} + \eta^2\, \omega^2\right)}
{\rho\, \left(G_0 + \sqrt{G_0^{\;2} + \eta^2\, \omega^2}\right)}}
\end{gather}
Both $\mathcal{A}_1(\omega)$ and $c(\omega)$ have a $\omega^{1/2}$ behavior in 
the high-frequency range,
while $\mathcal{A}_1(\omega) \sim_0 \left[ \rho^{1/2}\, \eta\right]/
\left[2 \ G_0^{\;3/2}\right] \omega^2$, $c(\omega) \sim_0 c_0$ for $\omega \rightarrow 0$.

In \cite{UrbanGreenleaf09} exponential power law attenuation was matched 
with Voigt models in the frequency range 50--600 Hz. 
For $\rho = 1000\, \mathrm{kg}\,\mathrm{m}^{-3}$, $G_0 = 2.5$  and $10\,\mathrm{kPa}$,
$\eta = 0.5, 2, 4\, \mathrm{Pa\, s}$ a reasonable fit was reported 
with exponent values close to $1.5$.
 It is however pointed out in \cite{MobleyWatersMiller03}  
that the frequency range covered in data acquisition available with
current technology is not adequate for demonstrating that the attenuation function 
satisfies a power law with an exponent in the range between 1 and 2.

Let $\beta(\omega) := \ln[\mathcal{A}_1(\omega)/\mathcal{A}_1(1)]/\ln(\omega)$ 
denote the local exponent of the attenuation function. $\beta$ is an analytic function.
Consequently if $\beta(\omega)$ is constant over a finite range of frequencies then
it is constant over the entire frequency range $0 < \omega < \infty$ and the attenuation
function $\mathcal{A}_1$ satisfies a power law. In Figs~\ref{fig:3}--\ref{fig:5} the 
local exponent has been plotted for the data of Ref.\cite{UrbanGreenleaf09} 
for the frequency range $50\div 600\,\mathrm{Hz}$ as well as for the
corresponding kilo-Hertz and megahertz ranges. Approximately constant 
exponents are observed over frequency ranges of logarithmic width $ \ln(50)$.
The solid line represents the exponent for $G_0 = 10\, \mathrm{kPa}$ and 
$\eta = 0.5 \,\mathrm{Pa\,s}$, the dashed line corresponds to 
$G_0 = 10\, \mathrm{kPa}$ and 
$\eta = 4 \,\mathrm{Pa\,s}$, and $G_0 = 2.5\, \mathrm{kPa}$ and 
$\eta = 0.5 \,\mathrm{Pa\,s}$.
\begin{figure}
\includegraphics[width=0.75\textwidth]{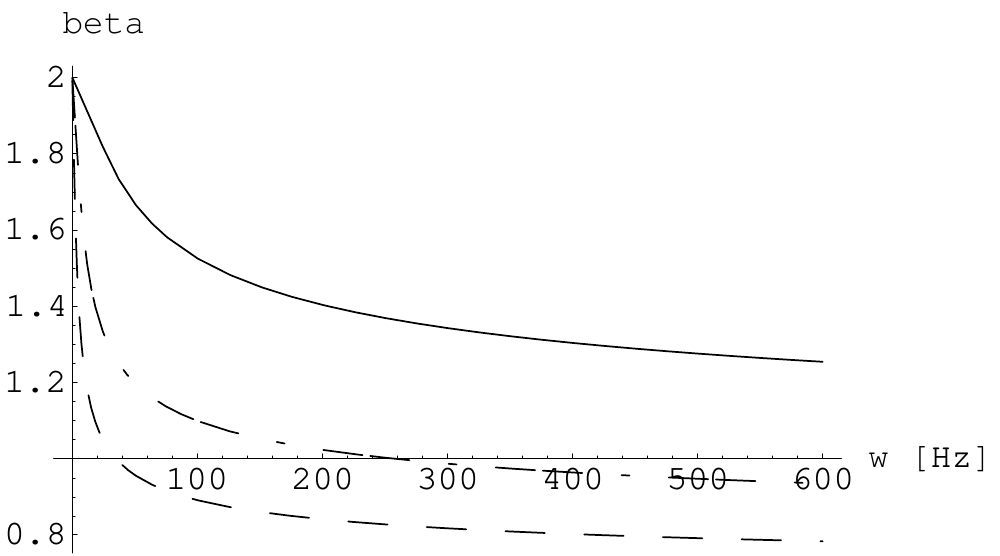}
\begin{center}
\caption{Local exponent in the frequency range $0\div 600~\mathrm{Hz}$.} \label{fig:3}
\end{center}
\end{figure}
\begin{figure}
\includegraphics[width=0.75\textwidth]{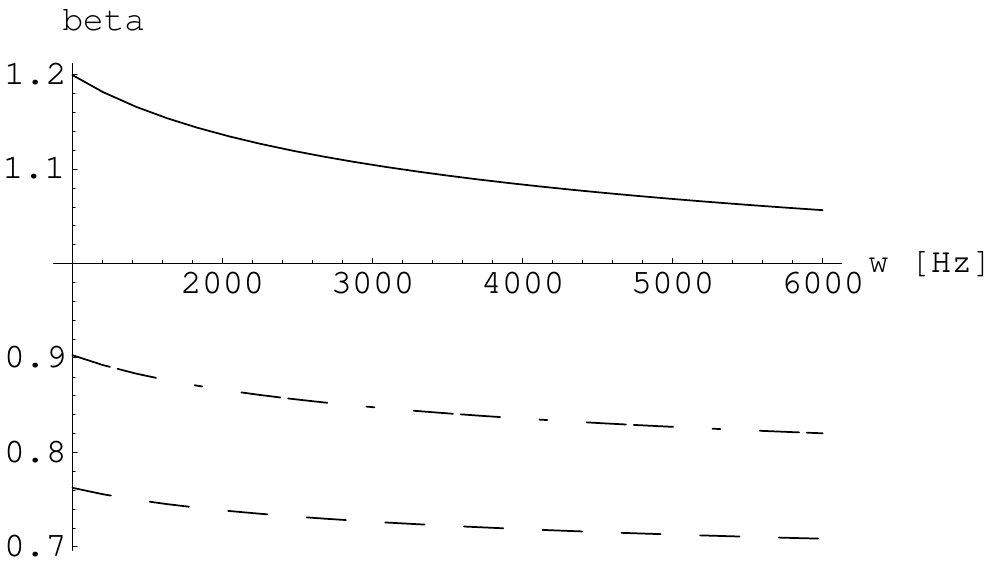}
\begin{center}
\caption{Local exponent in the frequency range $1 \div 6 \, \mathrm{kHz}$.} \label{fig:4}
\end{center}
\end{figure}
\begin{figure}
\includegraphics[width=0.75\textwidth]{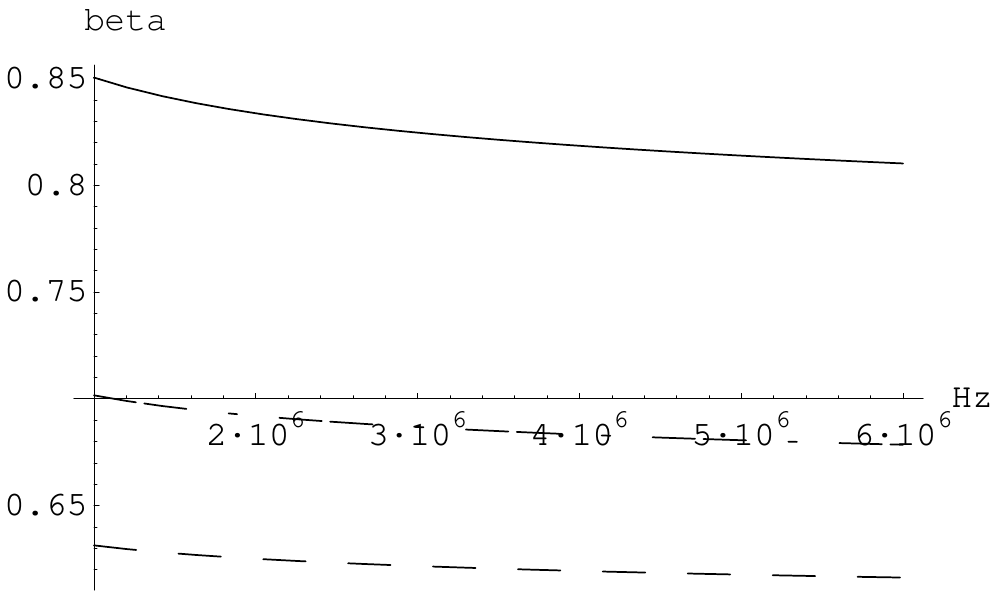}
\begin{center}
\caption{Local exponent in the frequency range $1 \div 6 \, \mathrm{MHz}$.}\label{fig:5}
\end{center}
\end{figure}

\section{Conclusions.}

The attenuation and dispersion functions of a viscoelastic medium with 
positive relaxation spectrum have integral representations in terms of a 
positive Radon measure. Each positive Radon measure satisfying the growth
condition \eqref{eq:5} determines an admissible attenuation function
and the corresponding dispersion function. The attenuation function
has a sublinear attenuation growth for $\omega \rightarrow \infty$. 

Superlinear growth of logarithmic attenuation rate 
for $\omega \rightarrow \infty$ is incompatible with hereditary viscoelastic constitutive 
equations with positive relaxation spectrum. It is also incompatible with 
finite propagation speed. Superlinear 
growth of attenuation results in a precursor of infinite extent preceding the main signal.

Superlinear dependence of attenuation on frequency observed in some viscous
materials is possible in the intermediate frequency range. 

\bibliography{mrabbrev,mathnew12,ownnew12}

\end{document}